\documentclass{article}

\usepackage{amsmath}
\usepackage{amssymb}
\usepackage{mdframed}
\usepackage{amsthm}
\usepackage{mathtools}
\usepackage{placeins}
\usepackage{url}

\newtheorem{thm}{Theorem}[section]
\newtheorem{obs}[thm]{Observation}
\newtheorem{lem}[thm]{Lemma}

\theoremstyle{remark}

\newtheorem*{rem}{Remark}

\DeclareMathOperator{\Dim}{Dim}
\DeclareMathOperator{\dimH}{dim_H}
\DeclareMathOperator{\dimP}{dim_P}
\DeclareMathOperator{\proj}{proj}
\DeclareMathOperator{\sgn}{sgn}
\newcommand{\R}{\mathbb{R}}

\newcommand{\N}{\mathbb{N}}
\newcommand{\Q}{\mathbb{Q}}
\newcommand{\ve}{\varepsilon}
\newcommand{\uhr}{{\upharpoonright}}

\title{Projection Theorems Using Effective Dimension}
	\author{
	Neil Lutz\\
	Computer Science Department, Swarthmore College\\
			500 College Avenue,
			Swarthmore, PA 19081\\
	\texttt{lutz@cs.swarthmore.edu}
	\and
	D. M. Stull\footnote{Research supported in part by National Science Foundation Grants 1247051 and 1545028.}\\
	Department of Computer Science, Northwestern University\\
	Evanston, IL 50208, USA\\
	\texttt{donald.stull@northwestern.edu}
}	
		
\begin{document}

		\maketitle
		
		\begin{abstract}
		In this paper we use the theory of computing to study fractal dimensions of projections in Euclidean spaces. A fundamental result in fractal geometry is Marstrand's projection theorem, which states that for every analytic set $E$, for almost every line $L$, the Hausdorff dimension of the orthogonal projection of $E$ onto $L$ is maximal.
		
		We use Kolmogorov complexity to give two new results on the Hausdorff and packing dimensions of orthogonal projections onto lines. The first shows that the conclusion of Marstrand's theorem holds whenever the Hausdorff and packing dimensions agree on the set $E$, even if $E$ is not analytic. Our second result gives a lower bound on the packing dimension of projections of arbitrary sets. Finally, we give a new proof of Marstrand's theorem using the theory of computing.
		\end{abstract}
		
	\section{Introduction}\label{sec:intro}
The field of fractal geometry studies the fine-grained structure of irregular sets. Of particular importance are fractal dimensions, especially the Hausdorff dimension, $\dimH(E)$, and packing dimension, $\dimP(E)$, of sets $E\subseteq\R^n$. Intuitively, these dimensions are alternative notions of size that allow us to quantitatively classify sets of measure zero. The books of Falconer~\cite{Falc14} and Mattila~\cite{Matt99} provide an excellent introduction to this field.

A fundamental problem in fractal geometry is determining how projection mappings affect dimension~\cite{FalFraJin15, Mattila14}. Here we study orthogonal projections of sets onto lines. Let $e$ be a point on the unit $(n-1)$-sphere $S^{n-1}$, and let $L_e$ be the line through the origin and $e$. The \emph{projection of $E$ onto $L_e$} is the set
\[\proj_e E=\{e\cdot x:x\in E\}\,,\]
where $e\cdot x$ is the usual dot product, $\sum_{i=1}^n e_ix_i$, for $e=(e_1,\ldots,e_n)$ and $x=(x_1,\ldots,x_n)$.
We restrict our attention to lines through the origin because translating the line $L_e$ will not affect the Hausdorff or packing dimension of the projection.

Notice that $\proj_e E\subseteq\R$, so the Hausdorff dimension of $\proj_e E$ is at most $1$. It is also simple to show that $\dimH(\proj_e E)$ cannot exceed $\dimH(E)$~\cite{Falc14}. Given these bounds, it is natural to ask whether $\dimH(\proj_e E)=\min\{\dimH(E), 1\}$. Choosing $E$ to be a line orthogonal to $L_e$ shows that this equality does not hold in general. However, a fundamental theorem due to Marstrand~\cite{Marstrand54} states that, if $E\subseteq\R^2$ is a Borel set then for \emph{almost all} $e \in S^{1}$, the Hausdorff dimension of $\proj_e E$ is maximal. In fact, this theorem holds for the larger class of \emph{analytic sets}, which are continuous images of Borel sets. Subsequently, Mattila~\cite{Mattila75} showed that the conclusion of Marstrand's theorem also holds in higher-dimensional Euclidean spaces.
\begin{thm}[\cite{Marstrand54,Mattila75}]\label{thm:main1}
Let $E \subseteq \R^n$ be an analytic set with $\dimH(E) = s$. Then for almost every  $e \in S^{n-1}$, 
\[\dimH(\proj_e E) = \min\{s, 1\}\,.\]
\end{thm}
In recent decades, the study of projections have become increasingly central to fractal geometry~\cite{FalFraJin15}. The most prominent technique has been the potential theoretic approach of Kaufman~\cite{Kaufman68}. While this is a very powerful tool in studying the dimension of a set, it requires that the set be analytic. We will show that techniques from theoretical computer science can circumvent this requirement in some cases.

Our approach to this problem is rooted in the effectivizations of Hausdorff dimension~\cite{Lutz03a} by J. Lutz and of packing dimension by Athreya et al.~\cite{AHLM07}. The original purpose of these effective dimension concepts was to quantify the size of complexity classes, but they also yield geometrically meaningful definitions of dimension for \emph{individual points} in $\R^n$~\cite{LutMay08}. More recently, J. Lutz and N. Lutz established a bridge from effective dimensions back to classical fractal geometry by showing that the Hausdorff and packing dimensions of a set $E \subseteq \R^n$ are characterized by the corresponding effective dimensions of the individual points in $E$, taken relative to an appropriate oracle~\cite{LutLut18}.

This result, a \emph{point-to-set principle} (Theorem~\ref{thm:p2s} below), allows researchers to use tools from algorithmic information theory to study problems in classical fractal geometry. Although this connection has only recently been established, there have been several results demonstrating the usefulness of the point-to-set principle: J. Lutz and N. Lutz~\cite{LutLut18} applied it to give a new proof of Davies' theorem~\cite{Davi71} on the Hausdorff dimension of Kakeya sets in the plane; N. Lutz and Stull~\cite{LutStu20} applied it to the dimensions of points on lines in $\R^2$ to give improved bounds on generalized Furstenberg sets; and N. Lutz~\cite{Lutz21} used it to show that a fundamental bound on the Hausdorff dimension of intersecting fractals holds for arbitrary sets.

In this paper, we use algorithmic information theory, via the point-to-set principle, to study the Hausdorff and packing dimensions of orthogonal projections onto lines. Given the statement of Theorem~\ref{thm:main1}, it is natural to ask whether the requirement that $E$ is analytic can be removed. Without further conditions, it cannot; Davies~\cite{Davies79} showed that, assuming the continuum hypothesis, there are non-analytic sets for which Theorem~\ref{thm:main1} fails. Indeed, Davies constructed a set $E\subseteq \R^2$ such that $\dimH(E)=1$ but $\dimH(\proj_e E)=0$ for \textit{every} $e \in S^1$.

Our first main theorem shows that if the Hausdorff and packing dimensions of $E$ agree, then we can remove the requirement that $E$ is analytic.
\begin{thm}\label{thm:main2}
Let $E \subseteq \R^n$ be any set with $\dimH(E) = \dimP(E) = s$. Then for almost every $e \in S^{n-1}$, 
\[\dimH(\proj_e E) = \min\{s, 1\}\,.\]
\end{thm}

Our second main theorem applies to projections of \emph{arbitrary} sets. Davies' construction precludes any non-trivial lower bound on the Hausdorff dimension of projections of arbitrary sets, but we are able to give a lower bound on the \emph{packing} dimension.
\begin{thm}\label{thm:main3}
Let $E \subseteq \R^n$ be any set with $\dimH(E) = s$. Then for almost every $e \in S^{n-1}$, 
\[\dimP(\proj_e E) \geq \min\{s, 1\}\,.\] 
\end{thm}
Lower bounds on the packing dimension of projections has been extensively studied for restricted classes of sets such as Borel and analytic sets~\cite{FH96, FH97, FM96, Howroyd01, Orp15}. This is the first non-trivial lower bound of this type for arbitrary sets. It is known that the direct analogue of Marstrand's theorem for packing dimension does \emph{not} hold~\cite{Jarv94}.

'Our other contribution is a new proof of Marstrands projection theorem Theorem \ref{thm:main1}. In addition to showing the power of theoretical computer science in geometric measure theory, this proof introduces a new technique for further research in this area. We show that the assumption that $E$ is analytic allows us to use an earlier, restricted point-to-set principle due to J. Lutz~\cite{Lutz03a} and Hitchcock~\cite{Hitc05}. While less general than that of J. Lutz and N. Lutz, it is sufficient for this application and involves a simpler oracle. Informally, this allows us to reverse the order of quantifiers in the statement of Theorem \ref{thm:main1}. This will be beneficial for further research in this area, and it is a step toward clarifying the role of the assumption that $E$ is analytic.

Finally, we remark that our three main theorems have natural generalizations to projecting onto higher dimensional subspaces and that these generalizations are achievable via straightforward modifications of the proofs given here. However, we restrict ourselves to lines for the sake of readability.

\section{Preliminaries}\label{sec:prelim}

We begin with a brief description of algorithmic information quantities and their relationships to Hausdorff and packing dimensions.

\subsection{Kolmogorov Complexity in Discrete and Continuous Domains}
The \emph{conditional Kolmogorov complexity} of binary string $\sigma\in\{0,1\}^*$ given a binary string $\tau\in\{0,1\}^*$ is the length of the shortest program $\pi$ that will output $\sigma$ given $\tau$ as input. Formally, the conditional Kolmogorov complexity of $\sigma$ given $\tau$ is
\[K(\sigma\mid\tau)=\min_{\pi\in\{0,1\}^*}\left\{\ell(\pi):U(\pi,\tau)=\sigma\right\}\,,\]
where $U$ is a fixed universal prefix-free Turing machine and $\ell(\pi)$ is the length of $\pi$. Any $\pi$ that achieves this minimum is said to \emph{testify} to, or be a \emph{witness} to, the value $K(\sigma\mid\tau)$. The \emph{Kolmogorov complexity} of a binary string $\sigma$ is $K(\sigma)=K(\sigma\mid\lambda)$, where $\lambda$ is the empty string.	These definitions extend naturally to other finite data objects, e.g., vectors in $\Q^n$, via standard binary encodings; see~\cite{LiVit08} for details.

Kolmogorov complexity can be further extended to points in Euclidean space, as we now describe. The \emph{Kolmogorov complexity} of a point $x\in\R^m$ at \emph{precision} $r\in\N$ is the length of the shortest program $\pi$ that outputs a \emph{precision-$r$} rational estimate for $x$. Formally, this is 
\[K_r(x)=\min\left\{K(p)\,:\,p\in B_{2^{-r}}(x)\cap\Q^m\right\}\,,\]
where $B_{\ve}(x)$ denotes the open ball of radius $\ve$ centered on $x$. The \emph{conditional Kolmogorov complexity} of $x$ at precision $r$ given $y\in\R^n$ at precision $s\in\R^n$ is
\[K_{r,s}(x\mid y)=\max\big\{\min\{K_r(p\mid q)\,:\,p\in B_{2^{-r}}(x)\cap\Q^m\}\,:\,q\in B_{2^{-s}}(y)\cap\Q^n\big\}\,.\]
When the precisions $r$ and $s$ are equal, we abbreviate $K_{r,r}(x\mid y)$ by $K_r(x\mid y)$. As a matter of notational convenience, if we are given a non-integral positive real as a precision parameter, we will always round up to the next integer. For example, $K_{r}(x)$ denotes $K_{\lceil r\rceil}(x)$ whenever $r\in(0,\infty)$.

By letting the underlying fixed prefix-free Turing machine $U$ be a universal \emph{oracle} machine, we may \emph{relativize} the definitions in this section to an arbitrary oracle set $A \subseteq \N$. The definitions of $K^A(\sigma\mid \tau)$, $K^A(\sigma)$, $K^A_r(x)$, $K^A_r(x\mid y)$, $\dim^A(x)$, $\Dim^A(x)$  $\dim^A(x\mid y)$, and $\Dim^A(x\mid y)$ are then all identical to their unrelativized versions, except that $U$ is given oracle access to $A$. We will frequently consider the complexity of a point $x \in \R^n$ \emph{relative to a point} $y \in \R^m$, i.e., relative to an oracle set $A_y$ that encodes the binary expansion of $y$ is a standard way. We then write $K^y_r(x)$ for $K^{A_y}_r(x)$.

One of the most useful properties of Kolmogorov complexity is that it obeys the \emph{symmetry of information}. That is, for every $\sigma, \tau \in \sigma\in\{0,1\}^*$,
\[K(\sigma, \tau) = K(\sigma) + K(\tau \mid \sigma, K(\sigma)) + O(1)\,.\]

We will need the following technical lemmas which show that versions of the symmetry of information hold for Kolmogorov complexity in $\R^n$. The first Lemma~\ref{lem:unichain} was proved in our previous work~\cite{LutStu20}.
\begin{lem}[\cite{LutStu20}]\label{lem:unichain}
	For every $m,n\in\N$, $x\in\R^m$, $y\in\R^n$, and $r,s\in\N$ with $r\geq s$,
	\begin{enumerate}
		\item[\textup{(i)}]$\displaystyle |K_r(x\mid y)+K_r(y)-K_r(x,y)\big|\leq O(\log r)+O(\log\log \|y\|)\,.$
		\item[\textup{(ii)}]$\displaystyle |K_{r,s}(x\mid x)+K_s(x)-K_r(x)|\leq O(\log r)+O(\log\log\|x\|)\,.$
	\end{enumerate}
Moreover, this proof relativizes, and so both (i) and (ii) hold relative to any oracle.
\end{lem}
	
\begin{lem}\label{lem:symmetry}
	Let $m,n\in\N$, $x\in\R^m$, $z\in\R^n$, $\ve > 0$ and $r\in\N$. If $K^x_r(z) \geq K_r(z) - \ve r$, then the following hold for all $s \leq r$.
	\begin{enumerate}
		\item[\textup{(i)}]$\displaystyle |K^x_s(z) - K_s(z)|\leq \ve r - O(\log r)\,.$
		\item[\textup{(ii)}]$\displaystyle |K_{s, r}(x \mid z) - K_s(x)|\leq \ve r - O(\log r)\,.$
	\end{enumerate}
\end{lem}

\begin{proof}
We first prove item (i). By Lemma \ref{lem:unichain}(ii),
\begin{align*}
\ve r &\geq K_r(z) - K^x_r(z)\\
&\geq K_s(z) + K_{r, s}(z \mid z) - (K^x_s(z) + K^x_{r, s}(z \mid z)) - O(\log r)\\
&\geq K_s(z) - K^x_s(z)+ K_{r, s}(z \mid z) - K^x_{r, s}(z \mid z) - O(\log r)\,.
\end{align*}
Rearranging, this implies that 
\begin{align*}
K_s(z) - K^x_s(z) &\leq \ve r +K^x_{r, s}(z \mid z) - K_{r, s}(z \mid z) + O(\log r) \\
&\leq \ve r + O(\log r)\,,
\end{align*}
and the proof of item (i) is complete. 

To prove item (ii), by Lemma \ref{lem:unichain}(i) we have
\begin{align*}
\ve r &\geq K_r(z) - K_r(z \mid x)\\
&\geq K_r(z) - (K_r(z, x) - K_{r}(x)) - O(\log r)\\
&\geq K_r(z) - (K_r(z) +K_r(x \mid z) - K_{r}(x)) - O(\log r)\\
&= K_{r}(x) - K_r(x \mid z) - O(\log r)\,.
\end{align*}
Therefore, by Lemma \ref{lem:unichain}(ii),
\begin{align*}
K_s(x) - K_{s, r}(x \mid z) &= K_r(x) - K_{r, s}(x \mid x) - (K_r(x \mid z) - K_{r, s, r}(x \mid x, z))\\
&\leq \ve r + O(\log r) + K_{r, s, r}(x \mid x, z) - K_{r, s}(x \mid x) \\
&\leq \ve r + O(\log r)\,,
\end{align*}
and the proof is complete.
\end{proof}
When simultaneously handling complexities of finite bit strings and of a Euclidean point $x\in\R^n$, it is often convenient to deal with finite prefixes of the $x$'s binary expansion instead of $K$-minimizing rational points in neighborhoods of $x$. Let $x\uhr r$ denote the concatenation of the first $r$ bits in the binary expansions of $x$'s $n$ coordinates. It is straightforward to show that the $r$-dyadic point represented by $x\uhr r$ has complexity very close to that of the lowest-complexity rational point in $B_{2^{-r}}(x)$. In particular,
\begin{equation}\label{eq:trunc}
    K_r(x)=K(x\uhr r)+O(\log r).
\end{equation}
This approximate equivalence also holds for conditional and relative complexities, so we can move freely between the two representations while only incurring $O(\log r)$ error. See~\cite{LutStu20} for a thorough accounting of the relationship between these two representations.

\subsection{Effective Hausdorff and Packing Dimensions}
J. Lutz~\cite{Lutz03a} initiated the study of effective dimensions (also known as \emph{algorithmic dimensions}) by effectivizing Hausdorff dimension using betting strategies called~\emph{gales}, which generalize martingales. Subsequently, Athreya et al., defined effective packing dimension, also using gales~\cite{AHLM07}. Mayordomo showed that effective Hausdorff dimension can be characterized using Kolmogorov complexity~\cite{Mayo02}, and Athreya, Hitchcock, J. Lutz and Mayordomo~\cite{AHLM07} showed that effective packing dimension can also be characterized in this way. In this paper, we use these characterizations as definitions.
The \emph{effective Hausdorff dimension} and \emph{effective packing dimension} of a point $x\in\R^n$ are
\[\dim(x)=\liminf_{r\to\infty}\frac{K_r(x)}{r}\quad\text{and}\quad\Dim(x) = \limsup_{r\to\infty}\frac{K_r(x)}{r}\,.\]
Intuitively, these dimensions measure the density of algorithmic information in the point $x$. J. Lutz and N. Lutz~\cite{LutLut18} generalized these definitions by defining the \emph{lower} and \emph{upper conditional dimension} of $x\in\R^m$ given $y\in\R^n$ as
\[\dim(x\mid y)=\liminf_{r\to\infty}\frac{K_r(x\mid y)}{r}\quad\text{and}\quad\Dim(x\mid y) = \limsup_{r\to\infty}\frac{K_r(x\mid y)}{r}\,.\]

\subsection{The Point-to-Set Principle}

The following \emph{point-to-set principles} show that the classical notions of Hausdorff and packing dimension of a set can be characterized by the effective dimension of its individual points. The first point-to-set principle, for a restricted class of sets, was implicitly proven by J. Lutz~\cite{Lutz03a} and Hitchcock~\cite{Hitc05}. 

A set $E \subseteq \R^n$ is a $\mathbf{\Sigma}^0_2$ set if it is a countable union of closed sets. The computable analogue of $\mathbf{\Sigma}^0_2$ is the class $\Sigma_2^0$ of sets $E\subseteq \R^n$ such that there is a uniformly computable sequence $\{C_{i}\}_{i \in \N}$ satisfying
\[E = \bigcup\limits_{i = 0}^\infty C_{i}\,,\]
and each set $C_i$ is \emph{computably closed}, meaning that its complement is the union of a computably enumerable set of open balls with rational radii and centers. We will use the fact that every $\mathbf{\Sigma}^0_2$ set is $\Sigma_2^0$ relative to some oracle.
\begin{thm}[\cite{Lutz03a,Hitc05}]\label{thm:strongPointToSetDim}
Let $E \subseteq \R^n$ and $A \subseteq \N$ be such that $E$ is a $\Sigma_2^0$ set relative to $A$. Then
\[\dimH(E) = \sup\limits_{x \in E} \dim^A(x)\,.\]
\end{thm}

J. Lutz and N. Lutz~\cite{LutLut18} improved this result to show that the Hausdorff and packing dimension of \emph{any} set $E\subseteq\R^n$ is characterized by the corresponding effective dimensions of individual points, relativized to an oracle that is optimal for the set $E$.
\begin{thm}[Point-to-set principle~\cite{LutLut18}]\label{thm:p2s}
Let $n \in \N$ and $E \subseteq \R^n$. Then
\begin{align*}
\dimH(E) &= \adjustlimits\min_{A \subseteq \N} \sup_{x \in E} \dim^A(x), \text{ and}\\
\dimP(E) &= \adjustlimits\min_{A \subseteq \N} \sup_{x \in E} \Dim^A(x)\,.
\end{align*}
\end{thm}

\section{Bounding the Complexity of Projections}\label{sec:BoundComplexityProj}
In this section, we will focus on bounding the Kolmogorov complexity of a projected point \emph{at a given precision}. In Section \ref{sec:ProjThm}, we will use these results in conjunction with Theorem \ref{thm:p2s} to prove our main theorems.

We begin by giving intuition of the main idea behind this lower bound. We will show that under certain conditions, given an approximation of $e\cdot z$ and $e$, we can compute an approximation of the original point $z$.  Informally, these conditions are the following.
\begin{enumerate}
\item The complexity $K_r(z)$ of the original point is small.
\item If $e\cdot w = e\cdot z$, then either $K_r(w)$ is large, or $w$ is close to $z$.
\end{enumerate}
Assuming that both conditions are satisfied, we can recover $z$ from $e\cdot z$ by enumerating over all points $u$ of low complexity such that $e\cdot u = e\cdot z$. By our assumption, any such point $u$ must be a good approximation of $z$. We now formally state this lemma.
\begin{lem}\label{lem:point}
Suppose that $z\in\R^n$, $e \in S^{n-1}$, $r\in\N$, $\delta\in\R_+$, and $\ve,\eta\in\Q_+$ satisfy $r\geq \log(2\|z\|+5)+1$ and the following conditions.
\begin{itemize}
\item[\textup{(i)}]$K_r(z)\leq \left(\eta+\ve\right)r$.
\item[\textup{(ii)}] For every $w \in B_1(z)$ such that $e\cdot w=e\cdot z$, \[K_{r}(w)\geq\left(\eta-\ve\right)r+(r- t)\delta\,,\]
whenever $t=-\log\|z-w\|\in(0,r]$.
\end{itemize}
Then for every oracle set $A\subseteq\N$,
\[K_r^{A, e}(e\cdot z)\geq K_r^{A,e}(z)-\frac{n\ve}{\delta}r-K(\ve)-K(\eta)-O(\log r)\,,\]
where the constant implied by the big oh bound depend only on $z, e,$ and $n$ .
\end{lem}

Our proof will use the following geometric observation, which is verified by routine calculations.

\begin{obs}\label{obs:existProj}
Let $z \in \R^n$, $p \in \Q^n$, $e \in S^{n-1}$, and $r \in \N$ such that $\vert e \cdot z  - e\cdot p \vert \leq 2^{-r}$. Then there is a $w \in \R^n$ such that $\| p - w \| \leq 2^{- r}$ and $e \cdot z  = e \cdot w $.
\end{obs}
\begin{proof}
Suppose $z \in \R^n$, $p \in \Q^n$, $e \in S^{n-1}$, and $r \in \N$ satisfy the hypothesis. Define
\[w = p + [e \cdot (z - p)]e\,.\]
Then we deduce that
\begin{align*}
    \| p - w\| &= \| p - (p + [e \cdot (z - p)]e)\|\\
    &= \|[e \cdot (z - p)]e\|\\
    &= \vert e \cdot (z-p)\vert \\
    &=\vert e \cdot z  - e\cdot p \vert\\
    &\leq 2^{-r}.
\end{align*}
We can also deduce that
\begin{align*}
    e \cdot w &= e \cdot p + e \cdot [e\cdot (z-p)]e\\
    &= e \cdot p + [e\cdot (z-p)](e \cdot e)\\
    &= e \cdot p + e\cdot (z-p)\\
    &= e \cdot z.
\end{align*}
\end{proof}

\begin{proof}[Proof of Lemma~\ref{lem:point}]
Suppose $z$, $e$, $r$, $\delta$, $\ve$, $\eta$, and $A$ satisfy the hypothesis.
				
Define an oracle Turing machine $M$ that does the following given oracle $(A,e)$ and input $\pi=\pi_1\pi_2\pi_3\pi_4\pi_5$ such that $U^A(\pi_1)=q\in\Q$, $U(\pi_2)=h\in\Q^n$, $U(\pi_3)=s\in\N$, $U(\pi_4)=\zeta\in\Q$, and $U(\pi_5)=\iota\in\Q$.
				
For every program $\sigma\in\{0,1\}^*$ with $\ell(\sigma)\leq (\iota+\zeta)s$, in parallel, $M$ simulates $U(\sigma)$. If one of the simulations halts with some output $p = (p_1,\ldots, p_n)\in \Q^n\cap B_{2^{-1}}(h)$ such that $|e\cdot p  - q|< 2^{-s}$, then $M^{A, e}$ halts with output $p$. Let $c_M$ be a constant for the description of $M$.
				
Let $\pi_1$, $\pi_2$, $\pi_3$, $\pi_4$, and $\pi_5$ testify to $K^{A,e}_r(e \cdot z )$, $K_1(z)$, $K(r)$, $K(\ve)$, and $K(\eta)$, respectively, and let $\pi=\pi_1\pi_2\pi_3\pi_4\pi_5$. Let $\sigma$ be a program of length at most $(\eta + \ve)r$ such that $\|p - z\| \leq 2^{-r}$, where $U(\sigma) = p$. Note that such a program must exist by condition (i) of our hypothesis. Since the function $x\mapsto e\cdot x$ is 1-Lipschitz,
\[|e \cdot z  - e\cdot p | \leq 2^{-r}\,,\]
for some fixed constant $c$ depending only on $z$ and $e$. Thus there is at least one program, $\sigma$, on which $M^{A, e}$ is guaranteed to halt on $\pi$.

Let $M^{A, e}(\pi) = p = (p_1,\ldots, p_n) \in \Q^n$. Observation~\ref{obs:existProj} shows that there is some
\[w \in B_{2^{\gamma-r}}(p)\subseteq B_{2^{-1}}(p)\subseteq B_{2^0}(z)\]
such that $e \cdot w  = e \cdot z$, where $\gamma$ is a constant depending only on $z$ and $e$.  Then,
\begin{align*}
K^{A,e}_r(w) &\leq |\pi| \\
&\leq K^{A, e}_r(e \cdot z ) + K_1(z) + K(r) + K(\ve) + K(\eta) + c_M\\
&= K^{A, e}_r(e \cdot z ) + K(\ve) + K(\eta) + O(\log r)\,,
\end{align*}
Rearranging this yields
\begin{equation}\label{eq:pointLowerPez}
K^{A, e}_r(e \cdot z ) \geq K^{A, e}_r(w) -  K(\ve) - K(\eta) - O(\log r)\,.
\end{equation}
Let $t = -\log\|z - w\|$. If $t \geq r$, then the proof is complete. If $t < r$, then $B_{2^{-r}}(p) \subseteq B_{2^{1-t}}(z)$, which implies that $K^{A, e}_r(w) \geq K^{A, e}_{t-1}(z)$. Therefore,
\begin{equation}\label{eq:pointLowerw}
K^{A, e}_r(w) \geq K^{A, e}_{r}(z) - n(r-t) - O(\log r)\,.
\end{equation}
We now bound $r - t$. By our construction of $M$,
\begin{align*}
(\eta + \ve) r &\geq K(p)\\
&\geq K_r(w) - O(\log r)\,.
\end{align*}
By condition (ii) of our hypothesis, then,
\[(\eta + \ve) r \geq (\eta - \ve) r + \delta(r - t)\,,\]
which implies that 
\[r - t \leq \frac{n\ve}{\delta} r + O(\log r)\,.\]
Combining this with inequalities~\eqref{eq:pointLowerPez} and~\eqref{eq:pointLowerw} concludes the proof.
\end{proof}

With the above lemma in mind, we wish to give a lower bound on the complexity of points $w$ such that $e\cdot w = e\cdot z$. Our next lemma gives a bound based on the complexity, relative to $z$, of the direction $e \in S^{n-1}$. This is based on the observation that we can calculate $e = (e_1,\ldots, e_n)$ given $w$, $z$, and $e_3,\ldots, e_n$, by solving the following pair of equations.
\begin{align}\label{eq:system}
\begin{split}
e\cdot(z - w)  &= 0 \\
e_1^2 + \ldots +e_n^2 &= 1\,.
\end{split}
\end{align}
This suggests that the complexity $K_r^{z, e_3,\ldots,e_n}(e)$ cannot be too much larger than $K^{z, e_3,\ldots,e_n}_r(w)$. However, for our purposes, we must be able to recover (an approximation of) $e$ given \emph{approximations} of $w$ and $z$. Intuitively, Lemma \ref{lem:lowerBoundOtherPoints} below shows that we can algorithmically compute an approximation of $e$ whose error is linearly correlated with distance between $w$ and $z$. We can then bound the complexity of $w$ using a symmetry of information argument.

\begin{lem}\label{lem:lowerBoundOtherPoints}
Let $z \in \R^n$, $e \in S^{n-1}$, and $r \in \N$. Let $w \in \R^n$ such that $\|z-w\|\geq 2^{-r}$, $e$ is nonzero in all coordinates, and $e\cdot z = e \cdot w $. Then there are numbers $1\leq i< j\leq n$ such that 
\begin{equation}\label{eq:lBOP}
    K_r(w) \geq K_t(z) + K^{e'}_{r-t, r}(e \mid z) + O(\log r)\,,
\end{equation}
where $t = -\log \|z - w\|$ and $e'\in\R^{n-2}$ consists of $e$ with its $i$\textsuperscript{th} and $j$\textsuperscript{th} coordinates removed.
\end{lem}
\begin{rem}
    When $n = 2$, we do not need the added complexity of removing coordinates. In this case, under the same hypothesis, we can conclude that
\begin{equation*}
    K_r(w) \geq K_t(z) + K_{r-t,r}(e\mid z) + O(\log r)\,.
\end{equation*}
This was proved in~\cite{LutStu20} in the context of lines.
\end{rem}
\begin{proof}
Let $z, w, e$, and $r$ be as in the statement of the lemma.

If $r-t=O(\log r)$, then $K_{r-t,r}(e\mid z)=O(\log r)$, so inequality~\eqref{eq:lBOP} holds trivially. Otherwise, Lemma~\ref{lem:quadratic} in the Appendix applies. It shows that there is an oracle Turing machine $M$ that uses the quadratic formula to find an approximate solution to the system of equations~\eqref{eq:system} given approximations of $z$ and $w$. More formally, if $q \in B_{2^{-r}}(z) \cap \Q^n$ and $\pi_q$ is a witness to
\[\hat{K}_r(w \mid q) = \min\{K(p \mid q) \, : \, p \in B_{2^{-r}}(w) \cap \Q^n\}\,,\]
then
\[\left\vert M^{e'}(\pi_q, q, h_1,h_2) - e\right\vert \leq 2^{\alpha+t -r}\,,\]
where $h_1,h_2\in\{0,1\}$, $\alpha$ is a constant depending only on $e$, and $e'$ is $e$ with some two of its coordinates removed. Thus, $K^{e'}_{r-t, r}(e \mid z) \leq \vert \pi_q\vert + O(1)$. Equivalently,
\begin{equation}\label{eq:eGivenW}
K^{e'}_{r-t, r}(e \mid z) \leq K_{r}(w \mid z) + O(1).
\end{equation}

To complete the proof, we apply the symmetry of information to note that 
\begin{align*}
K_{r}(w \mid z) &\leq K_{r, t}(w \mid z) + O(\log r) \\
&= K_{r,t}(w \mid w) + O(\log r)\\
&= K_r(w) - K_t(w) + O(\log r)\\
&= K_r(w) - K_t(z)+ O(\log r)\,.
\end{align*}
The lemma follows from rearranging the above inequality and combining it with inequality~\eqref{eq:eGivenW}.
\end{proof}

Finally, to satisfy the condition that $K_r(z)$ is small, we will use an oracle to ``artificially" decrease the complexity of $z$ at precision $r$. This requires some care, as we don't want the to affect the complexity of other points, or of $z$ at low precisions, any more than is necessary. The oracle we use will be based on following technical lemma from our earlier work.
\begin{lem}[\cite{LutStu20}]\label{lem:oracles}
Let $z\in\R^n$, $\eta\in\Q\cap[0,\dim(z)]$, and $r\in\N$. Then there is an oracle $D=D(r,z,\eta)$ with the following properties.
\begin{itemize}
\item[\textup{(i)}] For every $t\leq r$,
\[K^D_t(z)=\min\{\eta r,K_t(z)\}+O(\log r)\,.\]
\item[\textup{(ii)}] For every $m,t\in\N$ and $y\in\R^m$,
\[K^{D}_{t,r}(y\mid z)=K_{t,r}(y\mid z)+ O(\log r)\,,\]
and
\[K_t^{z,D}(y)=K_t^z(y)+ O(\log r)\,.\]
\end{itemize}
In particular, this oracle $D$ encodes $\sigma$, the lexicographically first time-minimizing witness to $K(z\uhr r\mid z\uhr s)$, where $s = \max\{t \leq r \, : \, K_{t-1}(z) \leq \eta r\}$, in the sense that, for all $\tau\in\{0,1\}^*$ and all oracles $B \subseteq \N$
\begin{equation}\label{eq:lftmw}
    K^{B,D}(\tau)=K^B(\tau\mid \sigma)+O(1)\,.
\end{equation}
\end{lem}

For our main theorems, we will need one additional property of the oracle $D$ in Lemma~\ref{lem:oracles}: For every oracle $B$, the amount that access to $B$ reduces the complexity of $z$ is not significantly increased by relativizing to $D$. The following lemma, which makes this precise, will allow us to take advantage of $D$'s properties, even in the presence of other oracles.
\begin{lem}\label{lem:independenceAeDr}
Let $z \in \R^n$, $B \subseteq \N$, $\eta \in \Q \cap [0, \dim(z)]$, $\ve > 0$, and $r \in \N$. Let $D=D(r,z,\eta)$ be the oracle defined in Lemma~\ref{lem:oracles}. If
\begin{equation}\label{eq:projthmassump}
    K^{B}_r(z) \geq K_r(z) - \ve r\,,
\end{equation}
then
\begin{equation}\label{eq:independenceAeDrGoal}
    K^{B, D}_r( z ) \geq K^{D}_r(z) - \ve r - O(\log r)\,.
\end{equation}
\end{lem}
\begin{proof}
Let $\sigma$ be as in Lemma~\ref{lem:oracles}. Then
\begin{equation*}
    K(\sigma \mid z\uhr r, s, K(z\uhr r\mid z\uhr s)) = O(1)\,.
\end{equation*}
Briefly, this is because, given the first $r$ bits of $z$, the value $K(z\uhr r\mid z\uhr s)$, and the parameter $s$, it is possible to execute an exhaustive search for the first time-minimizing witness to $K(z\uhr r\mid z\uhr s)$. Since $K(z\uhr r\mid z\uhr s)\leq r+O(\log r)$ and $s\leq r$, we have $K(K(z\uhr r\mid z\uhr s),s)=O(\log r)$, so by the symmetry of information,
\begin{equation}\label{eq:computeSigmaWithZ}
    K(\sigma \mid z\uhr r) = O(\log r)\,.
\end{equation}
This argument relativizes, so equation~\eqref{eq:computeSigmaWithZ} holds relative to all oracles. All that remains is a routine calculation using the symmetry of information and equation~\eqref{eq:trunc}:
\begin{align*}
    K^{B, D}_r( z ) &= K^{B, D}( z\uhr r )+O(\log r)\tag*{[equation~\eqref{eq:trunc}, relative to $(B,D)$]}\\
    &=K^{B}_r(z\uhr r\mid \sigma)+O(\log r)\tag*{[equation~\eqref{eq:lftmw}]}\\
    &= K^{B}(z\uhr r) +K^{B}(\sigma\mid z\uhr r) - K^{B}(\sigma) + O(\log r) \tag*{[symmetry of information]}\\
    &= K^{B}(z\uhr r) - K^{B}(\sigma) + O(\log r) \tag*{[equation~\eqref{eq:computeSigmaWithZ}, relative to $B$]}\\
    &=  K^{B}_r(z) - K^B(\sigma) + O(\log r) \tag*{[equation~\eqref{eq:trunc}, relative to $B$]}\\
    &\geq  K_r(z) - K^B(\sigma) - \ve r- O(\log r)\tag*{[inequality~\eqref{eq:projthmassump}]}\\
    &\geq K_r(z)  - K(\sigma) - \ve r - O(\log r)\tag*{[oracles cannot significantly increase complexity]}\\
    &= K(z\uhr r)  - K(\sigma) - \ve r - O(\log r)\tag*{[equation~\eqref{eq:trunc}]}\\
    &= K(z\uhr r)  + K(\sigma \mid z\uhr r) - K(\sigma) - \ve r - O(\log r)\tag*{[equation~\eqref{eq:computeSigmaWithZ}]}\\
    &= K(z\uhr r\mid \sigma)- \ve r - O(\log r)\tag*{[symmetry of information]}\\
    &= K^{D}(z\uhr r) - \ve r - O(\log r)\tag*{[equation~\eqref{eq:lftmw}]}\\
    &= K_r^{D}(z) - \ve r - O(\log r)\tag*{[equation~\eqref{eq:trunc}, relative to $D$]}\,,
\end{align*}
proving that~\eqref{eq:independenceAeDrGoal} holds.
\end{proof}

\section{Projection Theorems}\label{sec:ProjThm}
The main results of the previous section gave us sufficient conditions for strong lower bounds on the complexity of $e \cdot z $ at a given precision and methods to ensure that the conditions are satisfied. Theorem~\ref{thm:mainengine}  encapsulates these results so that we may apply them in the proofs of our main theorems on projections.

\begin{thm}\label{thm:mainengine}
Let $z \in \R^n$, $e \in S^{n-1}$ with all coordinates nonzero, $A \subseteq \N$, $\eta \in \Q \cap (0, \min\{\dim(z),1\})$, $\ve > 0$, and $r \in \N$. Assume the following conditions are satisfied.
\begin{enumerate}
\item For every $s \leq r$ and $1\leq i< j\leq n$, $K_{s}(e\mid e') \geq s - O(\log s)$.
\item $K^{A, e}_r(z) \geq K_r(z) - \ve r$,
\end{enumerate}
where $e'\in\R^{n-2}$ consists of $e$ with its $i$\textsuperscript{th} and $j$\textsuperscript{th} coordinates removed. Then,
\[K^{A, e}_r(e \cdot z ) \geq \eta r - \ve r -\frac{2n\ve}{1-\eta}r-K(\ve)-K(\eta)-O(\log r)\,.\]
\end{thm}

\begin{proof}
Assume the hypothesis. Let $D=D(r,z,\eta)$ be the oracle as defined in Lemma~\ref{lem:oracles}.

First assume that the conditions of Lemma \ref{lem:point}, relative to $D$, hold for $z, e$, $r$, $\eta$, $2\ve$, and $\delta = 1 - \eta$. Then, by Lemmas \ref{lem:point}, \ref{lem:oracles} and \ref{lem:independenceAeDr},
\begin{align*}
K_r^{A,  e,D}(e \cdot z ) &\geq K_r^{A, e,D}(z)-\frac{2n\ve}{1-\eta}r-K^{D}(2\ve)-K^{D}(\eta)-O(\log r)\tag*{[Lemma~\ref{lem:point}, relative to $D$]}\\
&\geq K^{D}_r(z) - \ve r-\frac{2n\ve}{\delta}r-K(\ve)-K(\eta)-O(\log r)\tag*{[Lemma~\ref{lem:independenceAeDr}]}\\
&= \eta r - \ve r-\frac{2n\ve}{\delta}r-K(\ve)-K(\eta)-O(\log r)\tag*{[Lemma~\ref{lem:oracles}(i)]}\,.
\end{align*}
Note that all constants hidden by the $O(\log r)$ term depend only on $z$. Therefore, to complete the proof, it suffices to show that the conditions of Lemma \ref{lem:point} hold relative to $D$. 

By property (i) of Lemma~\ref{lem:oracles},
\[K^D_r(z)\leq \varepsilon r+O(\log r)\,,\]
so for all sufficiently large $r$,
\[K^D_r(z)\leq (\eta +\ve)r\,,\]
which is condition (i) of Lemma \ref{lem:point}, relative to $D$. To see that condition (ii) of Lemma \ref{lem:point} also holds relative to $D$, let $w \in B_1(z)$ such that $e \cdot w =e \cdot z $ and $\|z-w\|\geq 2^{-r}$, and let $t=-\log\|z-w\|$. Then by Lemma \ref{lem:lowerBoundOtherPoints}, relative to $D$, there is some pair $1\leq i<j\leq n$ such that
\begin{align*}
K^{D}_r(w) &\geq K^{D}_t(z) + K^{D, e'}_{r-t, r}(e \mid z) + O(\log r)\\
&=K^{D}_t(z)+K^{e'}_{r-t,r}(e \mid z) - O(\log r)\tag*{[Lemma \ref{lem:oracles}(ii)]}\\
&\geq K^{D}_t(z)+K^{e'}_{r-t, r}(e) - \ve r - O(\log r)\tag*{[condition 2]}\\
&\geq \eta t + K^{e'}_{r-t, r}(e) - \ve r - O(\log r)\tag*{[Lemma \ref{lem:oracles}(i)]}\\
&\geq \eta t + r - t - \ve r - O(\log r)\tag*{[condition 1]}\\
&= t (\eta - 1) + r(1 - \ve) - O(\log r)\tag*{[rearranging]}\\
&\geq (\eta - \ve) r + \delta(r - t)\,,
\end{align*}
Hence, the conditions of Lemma \ref{lem:point} hold relative to $D$, and the proof is complete.
\end{proof}

\subsection{Projection Theorems For Non-Analytic Sets}\label{ssec:OtherPT}
Our first main theorem shows that, if the Hausdorff and packing dimensions of $E$ are equal, the conclusion of Marstrand's theorem holds. Essentially this assumption guarantees, for every oracle-direction pair $(A, e)$, the existence of a point $z \in E$ such that 
\begin{enumerate}
    \item $\dim^{A,e}(z) \geq \dimH(E) - \ve$, and
    \item the oracle $(A,e)$ does not change the complexity $K_r(z)$ at almost every precision $r$.
\end{enumerate}
This allows us to use Theorem \ref{thm:mainengine} at all sufficiently large precisions $r$.
{\renewcommand{\thethm}{\ref{thm:main2}}
\begin{thm}
Let $E \subseteq \R^n$ be any set with $\dimH(E) = \dimP(E) = s$. Then for almost every $e \in S^{n-1}$, 
\[\dimH(\proj_e E) \geq \min\{s, 1\}\,.\]
\end{thm}
\addtocounter{thm}{-1}}
\begin{proof}
Let $E \subseteq \R^n$ be any set with $\dimH(E) = \dimP(E) = s$. By Theorem \ref{thm:p2s}, there is an oracle $B\subseteq \N$ such that
\begin{equation}\label{eq:Bdef}
    \sup_{z\in E}\dim^B(z)=\sup_{z\in E}\Dim^B(z)=s\,.
\end{equation}
Let  $e \in S^{n-1}$ be any point which is random relative to $B$, meaning that
\begin{equation}\label{eq:edef}
    K^B_r(e)\geq (n-1)r-O(1)\,.
\end{equation}
Note that almost every point in $S^{n-1}$ satisfies this requirement and that such an $e$ will be nonzero in all coordinates. By Theorem \ref{thm:p2s} and the definition of $\proj_e E$, there is an oracle $A\subseteq \N$ such that
\begin{equation}\label{eq:Adef}
    \dimH(\proj_e E)=\sup_{z\in E}\dim^A(e\cdot z)\,.
\end{equation}
Hence, it suffices to show that for every $\ve > 0$ there is a $z \in E$ such that
\[\dim^{A}(e \cdot z ) \geq \min\{s, 1\} - \ve\,.\]

To that end, let $\eta \in \Q \cap (0, \min\{s,1\})$ and $\ve > 0$. By Theorem \ref{thm:p2s}, \[s\leq\sup_{z\in E}\dim^{A,B,\ve}(z)\,,\]
so there is a point $z_{\ve} \in E$ such that
\begin{equation}\label{eq:zedef}
    \dim^{B, A,e}(z_{\ve})\geq s - \frac{\ve}{2}\,.
\end{equation}
We now show that the conditions of Theorem \ref{thm:mainengine}, relative to $B$, are satisfied by $z_\ve$, $e$, $A$, $\eta$, $\varepsilon$, and all sufficiently large $r \in \N$. To see that condition 1 of Theorem \ref{thm:mainengine} holds, let $s\leq r$, let $1\leq i<j\leq n$, and let $e'$ consist of $e$ with its $i$\textsuperscript{th} and j\textsuperscript{th} coordinates removed. Then
\begin{align*}
    K^B_s(e\mid e')&=K_s^B(e)-K^B_s(e')+O(\log s)\tag*{[Lemma~\ref{lem:unichain}]}\\
    &\geq (n-1)s - K_s^B(e') - O(\log s)\tag*{[inequality~\eqref{eq:edef}]}\\
    &\geq (n-1)s-(n-2)(s+O(\log s))-O(\log s)\tag*{[$e'\in\R^{n-2}$]}\\
    &=s-O(\log s)\,.
\end{align*} 
To see that condition 2 holds, observe that for all sufficiently large $r$,
\begin{align*}
    K^B_r(z_\ve)-K^{B,A,e}_r(z_\ve)&\leq\frac{\ve r}{2}+r\cdot\limsup_{r\to\infty}\frac{K^B_r(z_\ve)-K^{B,A,e}_r(z_\ve)}{r}\tag*{[definition of $\limsup$]}\\
    &\leq\frac{\ve r}{2}+r\cdot\left(\Dim^B(z_\ve)-\dim^{B,A,e}(z_\ve)\right)\tag*{[definitions of $\dim$ and $\Dim$]}\\
    &\leq\frac{\ve r}{2}+r\cdot\left(s-\dim^{B,A,e}(z_\ve)\right)\tag*{[equation~\eqref{eq:Bdef}]}\\
    &\leq\frac{\ve r}{2}+r\cdot(s-(s-\ve/2))\tag*{[equation~\eqref{eq:zedef}]}\\
    &=\ve r\,,
\end{align*}
so $K_r^{B,A,e}(z_\ve)\geq K^B_r(z_\ve)-\ve r$.
%
Hence, the conditions of Theorem \ref{thm:mainengine}, relative to $B$, are satisfied. 


Therefore,
\begin{align*}
\dim^A(e \cdot z_{\ve}) &\geq \dim^{A, B, e}(e \cdot z_{\ve}) \tag*{[oracles cannot increase dimension]}\\
&= \liminf\limits_{r \rightarrow \infty} \frac{K^{A, B, e}_r(e \cdot z_{\ve})}{r}\tag*{[definition of $\dim$]}\\
&\geq \liminf\limits_{r \rightarrow \infty} \frac{\eta r - \ve r -\frac{2n\ve}{1-\eta}r-K(\ve)-K(\eta)-O(\log r)}{r}\tag*{[Theorem~\ref{thm:mainengine}, relative to $B$]}\\
&= \eta - \ve - \frac{2n\ve}{1-\eta}\,.
\end{align*}
Letting $\eta$ approach $\min\{s,1\}$ and $\varepsilon$ approach 0, we see that
\[\sup\limits_{z\in E}\dim^A(e \cdot z) \geq \min\{s,1\}\,,\]
so by equation~\eqref{eq:Adef}, the proof is complete.
\end{proof}

Our second main theorem gives a lower bound for the \emph{packing} dimension of a projection for \emph{general} sets. The proof of this theorem again relies on the ability to choose, for every $(A, e)$, a point $z$ whose complexity is unaffected by access to the oracle $(A, e)$. This cannot be assumed to hold for every precision $r$. However, by Theorem \ref{thm:p2s}, we can show that this can be done for infinitely many precision parameters $r$.
{\renewcommand{\thethm}{\ref{thm:main3}}
\begin{thm}
Let $E \subseteq \R^n$ be any set with $\dimH(E) = s$. Then for almost every  $e \in S^{n-1}$, 
\[\dimP(\proj_e E) \geq \min\{s, 1\}\,.\]
\end{thm}
\addtocounter{thm}{-1}}
\begin{proof}
	Let $E \subseteq \R^n$ be any set with $\dimH(E) = s$. By Theorem \ref{thm:p2s}, there is an oracle $B \subseteq \N$ such that 
	\begin{equation}\label{eq:Bhausoracle}
	    \dimH(E)=\sup_{x\in E}\dim^B(x)
	\end{equation}
	and
	\begin{equation}\label{eq:Bpackoracle}
	    \dimP(E)=\sup_{z\in E}\Dim^B(z)\,.
	\end{equation}
	Let  $e \in S^{n-1}$ be random relative to $B$; note that almost every $e\in S^{n-1}$ satisfies this requirement and that such an $e$ will be nonzero in all coordinates. Applying Theorem~\ref{thm:p2s} again, let $A \subseteq \N$ be an oracle such that
	\begin{equation}\label{eq:Apackoracle}
	    \dimP(\proj_e E)=\sup_{x\in \proj_e(E)}\Dim^A(x)=\sup_{z\in E}\Dim^A(e\cdot z)\,.
    \end{equation}
    Then it suffices to show that for every $\ve > 0$ there is a $z \in E$ such that
	\[\Dim^{A}(e \cdot z ) \geq \min\{s, 1\} - \ve\,.\]
	
	To that end, let $\eta \in \Q \cap (0,1) \cap (0, s)$ and $\ve > 0$. By Theorem \ref{thm:p2s}, there is a $z_{\ve} \in E$ such that
	\begin{equation}\label{eq:thirdThmInd1}
	s - \frac{\ve}{4} \leq \dim^{A, B, e}(z_{\ve})  \leq \dim^B(z_{\ve}) \leq s\,.
	\end{equation}

	We now show that the conditions of Theorem \ref{thm:mainengine} are satisfied relative to $B$ for infinitely many $r \in \N$. To see that condition 1 of Theorem \ref{thm:mainengine} holds, let $s\leq r$, let $1\leq i<j\leq n$, and let $e'$ consist of $e$ with its $i$\textsuperscript{th} and j\textsuperscript{th} coordinates removed. Then
\begin{align*}
    K^B_s(e\mid e')&=K_s^B(e)-K^B_s(e')+O(\log s)\tag*{[Lemma~\ref{lem:unichain}]}\\
    &\geq (n-1)s - K_s^B(e') - O(\log s)\tag*{[$e$ is random relative to $B$]}\\
    &\geq (n-1)s-(n-2)(s+O(\log s))-O(\log s)\tag*{[$e'\in\R^{n-2}$]}\\
    &=s-O(\log s)\,.
\end{align*} 
To see condition 2 holds, we first note that by inequality~\eqref{eq:thirdThmInd1} and the definition of $\dim$, there are cofinitely many $r\in\N$ such that
	\begin{equation}\label{eq:1.3.1}
	sr  - \frac{\ve}{2}r \leq r\cdot\left(\dim^{A,B,e}(z_\ve)-\frac{\ve}{4}\right)\\
	\leq K^{A, B, e}_r(z_{\ve})
	\end{equation}
	and infinitely many $r\in\N$ such that
	\begin{equation}\label{eq:1.3.2}
    K^B_r(z_{\ve})\leq r\cdot\left(\dim^B(z_\ve)+\frac{\ve}{4}\right)\leq sr + \frac{\ve}{4} r\,.
	\end{equation}
    Furthermore, since oracles cannot significantly increase complexity,
    \begin{equation}\label{eq:1.3.3}
        K_r^{A,B,e}(z_\ve)\leq K^B_r(z_\ve)+O(1)
        \leq K^B_r(z_\ve)+\frac{\ve r}{4}\,,
    \end{equation}
    for all but finitely many $r\in\N$. Thus there are infinitely many $r\in\N$ such that inequalities~\eqref{eq:1.3.1},~\eqref{eq:1.3.2}, and~\eqref{eq:1.3.3} all hold. For all such $r$,
    \[K_r^{A,B,e}(z_\ve),K_r^B(z_\ve)\in\left[sr-\frac{\ve}{2}r,sr+\frac{\ve}{2}r\right]\,.\]
    so we have
    \[K_r^{A,B,e}(z_\ve)\geq K_r^B(z_\ve)-\ve r\,,\]
    which is condition 2 of Theorem~\ref{thm:mainengine}, relative to $B$.
    
	Hence, we can apply Theorem \ref{thm:mainengine}, relative to $B$, for infinitely many $r \in \N$, so
	\begin{align*}
 	\Dim^A(e \cdot z_{\ve}) &\geq \Dim^{A, B, e}(e \cdot z_{\ve})\tag*{[oracles cannot increase dimension]}\\
 	&= \limsup\limits_{r \rightarrow \infty} \frac{K^{A, B, e}_r(e \cdot z_{\ve})}{r}\tag*{[definition of $\Dim$]}\\
 	&\geq \limsup\limits_{r \rightarrow \infty} \frac{\eta r - \ve r -\frac{2n\ve}{1-\eta}r-K(\ve)-K(\eta)-O(\log r)}{r}\tag*{[Theorem~\ref{thm:mainengine}, relative to $B$]}\\
 	&= \eta - \ve - \frac{2n\ve}{1-\eta}\,.
 	\end{align*}
	Letting $\eta$ approach $\min\{s,1\}$ $\ve$ approach 0, we see that 
    \[\sup\limits_{z \in E} \Dim^{A}(e \cdot z) \geq \min\{s, 1\}\,,\]
	so by equation~\eqref{eq:Apackoracle}, the proof is complete
\end{proof}

\subsection{Marstrand's Projection Theorem}\label{ssec:MPT}
We now give a new, algorithmic information theoretic proof of Marstrand's projection theorem. Recall its statement:
{\renewcommand{\thethm}{\ref{thm:main1}}
\begin{thm}[\cite{Marstrand54,Mattila75}]
Let $E \subseteq \R^n$ be analytic with $\dimH(E) = s$. Then for almost every  $e \in S^{n-1}$, \[\dimH(\proj_e E) = \min\{s, 1\}\,.\]
\end{thm}
\addtocounter{thm}{-1}}
Note the order of the quantifiers. To use Theorem \ref{thm:p2s}, we must first choose a direction $e \in S^{n-1}$. We then must show that for every oracle $A\subseteq\N$ and every $\ve > 0$, there is some point $z \in E$ satisfying
\[\dim^A(e \cdot z ) \geq \dimH(E) - \ve\,.\]
In order to apply Theorem \ref{thm:mainengine}, we must guarantee that access to the oracle $(A, e)$ does not significantly change the complexity of $z$. To ensure this, we will use the point-to-set principle of J. Lutz and Hitchcock (Theorem \ref{thm:strongPointToSetDim}). While this result is less general than the principle of J. Lutz and N. Lutz, the oracle characterizing the dimension of a $\Sigma^0_2$ set is easier to work with. 

To take advantage of this, we use the following lemma.
\begin{lem}\label{lem:Fsigma}
Let $E \subseteq \R^n$ be analytic with $\dimH(E)=s$. Then there is a $\mathbf{\Sigma}^0_2$ set $F \subseteq E$ such that $\dimH(F) = s$.
\end{lem}
\begin{proof}
Davies~\cite{Davies52}, generalizing work by Besicovitch~\cite{Besicovitch52}, showed that if $E \subseteq \R^n$ is analytic, then for every $\ve \in(0,s]$, there is a compact subset $E_\ve \subseteq E$ such that $\dimH(E_\ve) = s-\ve$. Thus, the set \[F=\bigcup_{i=\lceil 1/s\rceil}^\infty E_{1/i}\]
is a $\mathbf{\Sigma}^0_2$ set with $\dimH(F)=s$.
\end{proof}

We will also use the following observation, which is a consequence of the well-known fact from descriptive set theory that $\Sigma$ classes are closed under computable projections.

\begin{obs}\label{obs:compClosedOracle}
Let $E \subseteq \R^n$ and $A \subseteq \N$ be such that $E$ is a $\Sigma_2^0$ set relative to $A$. Then for every $e \in S^{n-1}$, $\proj_e E$ is a $\Sigma_2^0$ set relative to $(A, e)$.
\end{obs}

Finally, we must ensure that $e$ does not significantly change the complexity of $z$. For this, we will use the following definition and theorem due to Calude and Zimand~\cite{CalZim10}. We rephrase their work in terms of points in Euclidean space. Let $n \in N$, $z \in \R^n$ and $e \in S^{n-1}$. We say that $z$ and $e$ are \emph{independent} if, for every $r \in \N$, $K^e_r(z) \geq K_r(z) - O(\log r)$ and $K^z_r(e) \geq K_r(e) - O(\log r)$.
\begin{thm}[\cite{CalZim10}]\label{thm:CalZimInd}
For every $z \in \R^n$, for almost every $e \in S^{n-1}$, $z$ and $e$ are independent.
\end{thm}

With these ingredients we can give a new proof Marstrand's projection theorem using algorithmic information theory.
\begin{proof}[Proof of Theorem \ref{thm:main1}]
Let $E \subseteq \R^n$ be analytic with $\dimH(E) = s$. By Lemma \ref{lem:Fsigma}, there is a $\mathbf{\Sigma}^0_2$ set $F \subseteq E$ such that $\dimH(F) = s$. Let $A \subseteq \N$ be an oracle such that $F$ is $\Sigma_2^0$ relative to $A$. Hence, we can apply Theorem~\ref{thm:strongPointToSetDim}, which tells us that
\[\dim_H(F)=\sup_{x\in F}\dim^A(x)\,.\]
In particular, for every $k \in \N$, we may choose a point $z_k \in F$ such that 
\begin{equation}\label{eq:zkchoice}
	\dim^A(z_k) \geq s - 1/k\,.
\end{equation}

Let $e \in S^{n-1}$ be a point such that, relative to $A$, for all $k\in\N$, $e$ is random relative to $z_k$ and $e$ and $z_k$ are independent. It is a basic fact of algorithmic randomness that almost every $e$ satisfies the former condition and that such an $e$ will be nonzero in all coordinates. Theorem~\ref{thm:CalZimInd}, taken relative to $A$, tells us that almost every $e$ also satisfies the latter condition.

Fix $k\in\N$, let $z=z_k$, let $\eta\in \Q \cap (0, \min\{\dim(z),1\})$, let $\varepsilon\in\Q\cap(0,1)$, and let $r$ be sufficiently large that the $O(\log r)$ term from Theorem~\ref{thm:CalZimInd} is less than $\varepsilon r$. We now show that our choices of $z$, $e$, $A$, $\eta$, $\varepsilon$, and $r$ satisfy the two conditions of Theorem~\ref{thm:mainengine}. For the first condition, let $t\leq r$, let $1\leq i<j\leq n$, and let $e'\in\R^{n-2}$ consist of $e$ with its $i$\textsuperscript{th} and $j$\textsuperscript{th} coordinates removed. Then
\begin{align*}
	K_t(e\mid e')&\geq K^{A,z}_t(e\mid e')-O(1)\tag*{[oracles cannot significantly increase complexity]}\\
	&\geq K^{A,z}_t(e)-K^{A,z}_t(e')-O(\log t)\tag*{[Lemma~\ref{lem:unichain}]}\\
	&\geq K^{A,z}_t(e)-(n-2)(t+O(\log t))-O(\log t)\tag*{[$e'\in\R^{n-2}$]}\\
	&= K^{A,z}_t(e)-(n-2)t-O(\log t)\tag*{[$n$ is a constant]}\\
	&\geq (n-1)t-O(1)-(n-2)t-O(\log t)\tag*{[$e$ is random relative to $(A,z)$]}\\
	&=t-O(\log t),
\end{align*}
so the first condition of Theorem~\ref{thm:mainengine} is satisfied. The second condition is also satisfied because we have
\begin{align*}
	K^{A, e}_r(z)&\geq K^A_r(z)-O(\log r)\tag*{[$e$ and $z$ are independent]}\\
	&> K^A_r(z)-\varepsilon r\,.\tag*{[$r$ was chosen to be sufficiently large]}
\end{align*}
Hence, we may apply Theorem \ref{thm:mainengine}.

We now have
\begin{align*}
\dim^{A, e}(e \cdot z) &= \liminf\limits_{r \rightarrow \infty} \frac{K^{A, e}_r(e \cdot z)}{r}\tag*{[definition of $\dim$]}\\
&\geq \liminf\limits_{r \rightarrow \infty} \frac{\eta r - \ve r -\frac{2n\ve}{1-\eta}r-K(2\ve)-K(\eta)-O_{z}(\log r)}{r}\tag*{[Theorem~\ref{thm:mainengine}]}\\
&= \eta - \ve - \frac{2n\ve}{1-\eta}\,.\tag*{[$\varepsilon$ and $\eta$ are constants]}
\end{align*}
Thus, letting, $\varepsilon$ approach 0 and $\eta$ approach
\begin{align*}
	\min\{\dim(z),1\}&\geq\min\{\dim^A(z),1\}\tag*{[oracles cannot increase dimension]}\\
	&\geq\min\{s-1/k,1\}\,,
\end{align*}
and recalling that $z=z_k$, yields
\begin{equation}\label{eq:dimez}
	\dim^{A,e}(e\cdot z_k)\geq\min\{s-1/k,1\}\,.
\end{equation}

This holds for all $k\in\N$, and Observation~\ref{obs:compClosedOracle} tells us that $\proj_e F$ is $\Sigma_0^2$ relative to $(A,e)$, so we have
\begin{align*}
\dimH(\proj_e E) &\geq \dimH(\proj_e F)\tag*{[$F\subseteq E$]}\\
&=\sup_{y\in \proj_e F}\dim^{A,e}(y)\tag*{[Theorem~\ref{thm:strongPointToSetDim}]}\\
&= \sup\limits_{x \in F} \dim^{A, e}(e \cdot x )\tag*{[definition of $\proj_e$]}\\
&\geq \min\{s-1/k,1\}\tag*{[inequality~\eqref{eq:dimez}]}\\
&= \min\{s, 1\}\,.
\end{align*}
Since this holds for almost every $e\in S^{n-1}$, the proof is complete.
\end{proof}

\section{Conclusion}
In this paper we studied how the Hausdorff dimension of a set is changed under orthogonal projections onto a line. This question and its generalizations are of fundamental importance in geometric measure theory. We gave a new proof of Marstrand's projection theorem and two extensions to that theorem. Recently, Orponen \cite{Orponen21} used combinatorial techniques to give new proofs of Theorems~\ref{thm:main2} and~\ref{thm:main3}. This is an exciting development which gives hope for future interactions between the algorithmic and classical perspectives of geometric measure theory.

The main methodological contribution of this paper is the technique used in our proof of Marstrand's projection theorem. Specifically, in order to translate that theorem into its pointwise analogue, we used the point-to-set principle for $\Sigma^0_2$ sets (Theorem~\ref{thm:strongPointToSetDim}), due to Hitchcock~\cite{Hitc05} and J. Lutz~\cite{Lutz03a}. This more agile point-to-set approach allowed us to prove a theorem that does \emph{not} hold for arbitrary sets, and we therefore expect it to have further applications in geometric measure theory.

It is likely that further progress into related questions can be made with the use of the algorithmic perspective. One natural direction is to develop a new, point-to-set proof of the Marstrand's \emph{second} projection theorem, which states that, if $E$ is analytic and $\dim_H(E) > 1$, then for almost every unit vector $e$, the projected set $\proj_e(E)$ has positive measure. Since Theorems~\ref{thm:strongPointToSetDim} and~\ref{thm:p2s} only address dimension, not measure, this would require some refinement or extension of the point-to-set principle.

\bibliography{MPT}

\appendix
\renewcommand{\thesection}{\Alph{section}}
\section{Appendix}

The following technical lemma, which is used in the proof of Lemma~\ref{lem:lowerBoundOtherPoints}, shows that, given access to $n-2$ coordinates of a unit vector $e\in S^{n-1}$, together with approximations of two points $w$ and $z$ that have the same projection in the direction of $e$, a Turing machine can apply the quadratic formula to approximate the remaining coordinates of $e$. Although this fact is geometrically unsurprising, verifying it requires some cumbersome---but essentially routine---calculations, due to the large number of parameters and approximations involved. We reiterate that this added complexity is only needed for the case $n\geq 3$. 

\begin{lem}\label{lem:quadratic}
There is an oracle Turing machine $M$ with the following property. Suppose $w,z\in\R^n$ and $e\in S^{n-1}$ are such that $t=-\log\|w-z\|\leq r-16-3\log n-\log|e_2|$, $e$ is nonzero in all coordinates, and $e\cdot w=e\cdot z$; $p = (p_1, \ldots, p_n) \in \Q^n\cap B_{2^{-r}}(z)$, $r\in\N$, and $\pi_q$ is a witness to 
\[\hat{K}_r(w \mid q) = \min\{K(p \mid q) \, : \, p \in B_{2^{-r}}(w) \cap \Q^n\}\,.\]
Then there are some $h_1,h_2\in\{0,1\}$ and some $1\leq i<j\leq n$ such that
\[\left|M^{e'}(\pi_q,q,h_1,h_2)-e\right|\leq 2^{\alpha+t-r}\,,\]
where $\alpha$ is a constant depending only on $e$ and $e'\in\R^{n-2}$ is $e$ with its $i$\textsuperscript{th} and $j$\textsuperscript{th} coordinates removed.

\end{lem}

\begin{proof}
Let $w$, $z$, $e$, $p$, $r$, and $\pi_q$ be as in the lemma statement. We first choose $i$ so that $|z_i - w_i|$ is maximal. We then choose $j$ so that
\begin{enumerate}
    \item $\sgn((z_i - w_i)e_i) \neq \sgn((z_j - w_j)e_j),$ and
    \item $|z_j - w_j| > 0$
\end{enumerate}
where $\sgn$ denotes the sign. Note that such a $j$ must exist since $(z - w) \cdot e = 0$. In order to remove notational clutter, we will assume, without loss of generality, that $i = 1$ and $j = 2$, so $e'=(e_3,\ldots,e_n)$.

The following system of equations holds by our choices of $e$, $z$, and $w$.
\begin{align*}
e\cdot(z - w)  &= 0 \\
e_1^2 + \ldots +e_n^2 &= 1\,,
\end{align*}
Thus we can solve for $e_2$ using the quadratic formula
\begin{equation}\label{eq:quadForm}
e_2 = \frac{-b + (-1)^{h_2} \sqrt{b^2 - 4ac}}{2a}\,,
\end{equation}
where 
\begin{itemize}
\item $\displaystyle h_2 \in \{0,1\}$,
\item $\displaystyle a = (w_1 - z_1)^2 + (w_2 - z_2)^2$,
\item $\displaystyle b = 2(w_2 - z_2)\sum_{k=3}^n (w_k - z_k)e_k$, and
\item $\displaystyle c = \left(\sum_{k=3}^n (w_k - z_k)e_k\right)^2 - (w_1 - z_1)^2 \left(1-\sum_{k=3}^n e_k^2\right)$.
\end{itemize}

With this in mind, let $M$ be the oracle Turing machine such that, whenever $q = (q_1, \ldots, q_n) \in \Q^n$ and $U(\pi, q) = p = (p_1,\ldots, p_n) \in \Q^n$ with $p_1 \neq q_1$, and $h_1,h_2\in\{0,1\}$, $M^{e'}(\pi, q, h_1,h_2)$ does the following.
\begin{itemize}
    \item Choose $\displaystyle d = (d_3,\ldots, d_n) \in \Q^{n-2}\cap B_{2^{-r}}(e')$.
    \item Calculate $\displaystyle \hat{a} = (q_1 - p_1)^2 + (p_2 - q_2)^2$.
    \item Calculate $\displaystyle \hat{b} = 2(p_2 - q_2)\sum_{k=3}^n (p_k - q_k)d_k$.
    \item Calculate $\displaystyle \hat{c} = \left(\sum_{k=3}^n (p_k - q_k)d_k\right)^2 - (p_1 - q_1)^2\left(1-\sum_{k=3}^n d_k^2\right)$.
    \item Calculate $\displaystyle \hat{e}_2=\frac{-\hat{b} + (-1)^{h_2} \sqrt{\hat{b}^2 - 4\hat{a} \hat{c}}}{2\hat{a}}$.
    \item Calculate $\displaystyle \hat{e}_1=(-1)^{h_1}\sqrt{1-\hat{e}_2^2-d_3^2-\cdots-d_n^2}$.
    \item Output $\displaystyle \hat{e}=(\hat{e}_1,\hat{e}_2,d_3,\ldots,d_n)$.
\end{itemize}

To show that $\hat{e}$ is a good approximation of $e$ (assuming the correct choices of $h_1$ and $h_2$), we begin by bounding $|a-\hat{a}|$:
\begin{align*}
	|a-\hat{a}|&=|(w_1 - z_1)^2 + (w_2 - z_2)^2-(p_1 - q_1)^2 - (p_2 - q_2)^2|\\
	&\leq |(w_1 - z_1)^2-(p_1 - q_1)^2|+|(w_2 - z_2)^2-(p_2 - q_2)^2|\,.
\end{align*}
We bound the first term:
\begin{align*}
	|(w_1 - z_1)^2-(p_1 - q_1)^2|&=|(w_1-z_1)+(p_1-q_1)|\cdot|(w_1-z_1)-(p_1-q_1)|\\
	&\leq (|w_1-z_1|+|p_1-q_1|)\cdot(|w_1-p_1|+|z_1-q_1|)\\
	&\leq (2|w_1-z_1|+|p_1-w_1|+|q_1-z_1|)\cdot(|w_1-p_1|+|z_1-q_1|)\\
	&< (2^{1-t}+2^{1-r})\cdot(2^{1-r})\\
	&\leq 2^{3-r-t}\,.
\end{align*}
By identical reasoning, $|(w_2 - z_2)^2-(p_2 - q_2)^2|<2^{3-r-t}$, so we have
\begin{equation}\label{eq:boundonaminusa}
    |a-\hat{a}|<2^{4-r-t}
\end{equation}
Similarly, we bound $|b-\hat{b}|$:
\begin{align*}
		|b-\hat{b}|&=2\left |(w_2-z_2)\sum_{k=3}^n (w_k - z_k)e_k-(p_2 - q_2)\sum_{k=3}^n (p_k - q_k)d_k\right|\\
		&\leq 2\left|(w_2-z_2)\sum_{k=3}^n \big((w_k - z_k)e_k-(p_k - q_k)d_k\big)\right|+2\left|\big((w_2-z_2)-(p_2-q_2)\big)\sum_{k=3}^n (p_k - q_k)d_k\right|\\
		&< 2^{1-t}\left|\sum_{k=3}^n \big((w_k - z_k)e_k-(p_k - q_k)d_k\big)\right|+2^{2-r}\left|\sum_{k=3}^n (p_k - q_k)d_k\right|\\
		&\leq 2^{1-t}\left|\sum_{k=3}^n \big((w_k - z_k)e_k-(p_k - q_k)d_k\big)\right|+2^{2-r}(n-2)(2^{-t}+2^{1-r})\\
		&< 2^{1-t}\left|\sum_{k=3}^n \big((w_k - z_k)e_k-(p_k - q_k)d_k\big)\right|+2^{4-r-t+\log n}\\
		&\leq 2^{1-t}\left(\left|\sum_{k=3}^n (w_k - z_k)(e_k-d_k)\right|+\left|\sum_{k=3}^n(w_k-z_k-p_k+q_k)d_k\right|\right)+2^{4-r-t+\log n}\\
		&< 2^{1-t}\left((n-2)(2^{-t})(2^{-r})+(n-2)(2^{1-r})\right)+2^{4-r-t+\log n}\\
		&< 2^{3-r-t+\log n}+2^{4-r-t+\log n}\\
		&<2^{5-r-t+\log n}\,.\stepcounter{equation}\tag{\theequation}\label{eq:boundonbminusb}
\end{align*}

The difference
\[|c-\hat{c}|=\left|\left(\sum_{k=3}^n (w_k - z_k)e_k\right)^2 - (w_1 - z_1)^2 \left(1-\sum_{k=3}^n e_k^2\right)-\left(\sum_{k=3}^n (p_k - q_k)d_k\right)^2 + (p_1 - q_1)^2\left(1-\sum_{k=3}^n d_k^2\right)\right|\]
is unwieldy, so we handle this in pieces. First,
\begin{align*}
    \left|\sum_{k=3}^n(w_k - z_k)e_k-(p_k - q_k)d_k\right|&\leq\sum_{k=3}^n\left|(w_k-z_k-p_k+q_k)d_k\right|+\sum_{k=3}^n\left|(w_k-z_k)(e_k-d_k)\right|\\
    &< (n-2)2^{1-r}+(n-2)2^{-t}2^{-r}\\
    &<2^{2-r+\log n}\,,
\end{align*}
and
\begin{align*}
    \left|\sum_{k=3}^n(w_k - z_k)e_k+(p_k - q_k)d_k\right|&\leq \left|\sum_{k=3}^n(w_k - z_k)\right|+\left|\sum_{k=3}^n(p_k - q_k)\right|\\
    &\leq (n-2)2^{-t}+(n-2)(2^{-t}+2^{1-r})\\
    &<2^{2-t+\log n}
\end{align*}
so we have
\begin{align*}
    \left|\left(\sum_{k=3}^n (w_k - z_k)e_k\right)^2-\left(\sum_{k=3}^n (p_k - q_k)d_k\right)^2\right|&<2^{2-r+\log n}\cdot 2^{2-t+\log n}\\
    &=2^{4-r-t+2\log n}\,.\stepcounter{equation}\tag{\theequation}\label{eq:boundoncminuscpt1}
\end{align*}
Second,
\begin{align*}
    \left|(w_1-z_1)^2\left(\left(1-\sum_{k=3}^n e_k^2\right)-\left(1-\sum_{k=3}^n d_k^2\right)\right)\right|&=\left|(w_1-z_1)^2\sum_{k=3}^n (e_k-d_k)(e_k+d_k)\right|\\
    &< 2^{-2t}(n-2)2^{-r}(2+2^{-r})\\
    &<2^{2-2t-r+\log n}\,,
\end{align*}
and
\begin{align*}
    \left|\left((w_1-z_1)^2-(p_1-q_1)^2\right)\left(1-\sum_{k=3}^n d_k^2\right)\right|
    &\leq \left|(w_1-z_1)^2-(p_1-q_1)^2\right|\\
    &= |w_1-z_1-p_1+q_1||w_1-z_1+p_1-q_1|\\
    &< (2^{1-r})(2^{1-t}+2^{1-r})\\
    &\leq 2^{3-r-t}\,,
\end{align*}
so we have
\begin{align*}
    \left|(w_1 - z_1)^2 \left(1-\sum_{k=3}^n e_k^2\right)-(p_1-q_1)^2\left(1-\sum_{k=3}^n d_k^2\right)\right|&< 2^{2-2t-r+\log n}+2^{3-r-t}\\
    &\leq 2^{3-r-t+\log n}\,.\stepcounter{equation}\tag{\theequation}\label{eq:boundoncminuscpt2}
\end{align*}
Combining inequalities~\eqref{eq:boundoncminuscpt1} and~\eqref{eq:boundoncminuscpt2} yields
\begin{align*}
    |c-\hat{c}|&<2^{4-r-t+2\log n}+ 2^{3-r-t+\log n}\\
    &<2^{5-r-t+2\log n}\,.\stepcounter{equation}\tag{\theequation}\label{eq:boundoncminusc}
\end{align*}





We now bound $a, b$ and $c$. By our assumption that $|z_1-w_1|$ is maximal and the fact that $\|w-z\|=2^{-t}$, we have
\begin{align*}
a &= (z_1 - w_1)^2 + (z_2 - w_2)^2 \\
&\geq (z_1 - w_1)^2\\
&\geq \frac{2^{-2t}}{n}\\
&= 2^{-2t - \log n}\,.
\end{align*}
We also have $a \leq \|w-z\|^2= 2^{-2t}$, resulting in 
\begin{equation}\label{eq:boundona}
2^{-2t-\log n} \leq a \leq 2^{-2t}\,.
\end{equation}
We can combine this with the bound on $\vert a - \hat{a}\vert$ to conclude
\[2^{-2t - \log n} - 2^{4-r-t} \leq \vert \hat{a} \vert\,.\]
Under the assumption that $t \leq r - 16 - 3\log n - \log |e_2|$, we conclude that
\[2^{-2t - \log n - 1} \leq \vert \hat{a} \vert\,.\]
Similarly, we bound $b$ by deducing
\begin{align*}
    |b|& = \left|2(w_2 - z_2)\sum_{k=3}^n (w_k - z_k)e_k\right|\\
    &\leq 2^{1-t}\left|\sum_{k=3}^n 2^{-t}\right|\\
    &=(n-2)\cdot 2^{1-2t}\\
    &<2^{1-2t+\log n}\stepcounter{equation}\tag{\theequation}\label{eq:boundonb}
\end{align*}
Finally, we bound $c$
\begin{align*}
    |c| &= \left|\left(\sum_{k=3}^n (w_k - z_k)e_k\right)^2 - (z_1 - w_1)^2\left(1-\sum_{k=3}^n e_k^2\right)\right|\\
    &\leq \left|\left(\sum_{k=3}^n (w_k - z_k)e_k\right)^2\right| + \left|(z_1 - w_1)^2\right|\\
    &\leq ((n-2)\cdot 2^{-t})^2+2^{-2t}\\
    &<2^{1-2t+2\log n}\,.\stepcounter{equation}\tag{\theequation}\label{eq:boundonc}
\end{align*}


We now show that $c < 0$. To see this, we first deduce that
\begin{align*}
c &= \left(\sum_{i=3}^n (w_i - z_i)e_i\right)^2 - (z_1 - w_1)^2\left(1-\sum_{i=3}^n e_i^2\right) \\
&= ((z_1 - w_1)e_1 + (z_2 - w_2)e_2)^2 - (z_1 - w_1)^2 + (z_1 - w_1)^2(1-e_1^2 - e_2^2)\\
&= ((z_1 - w_1)e_1 + (z_2 - w_2)e_2)^2  + (z_1 - w_1)^2(-e_1^2 - e_2^2)\\
&= 2(z_1 - w_1)(z_2 - w_2)e_1 e_2 + e_2^2((z_2 - w_2)^2 - (z_1-w_1)^2).
\end{align*}
Since $|z_1 - w_1|$ is maximal,
\begin{equation*}
    e_2^2((z_2 - w_2)^2 - (z_1-w_1)^2) \leq 0.
\end{equation*}
Since $\sgn((z_1 - w_1)e_1) \neq \sgn((z_2 - w_2)e_2)$, 
\begin{equation*}
    (z_1 - w_1)(z_2 - w_2)e_1 e_2 < 0.
\end{equation*}
Therefore the claim that $c < 0$ has been verified.

Since $c < 0$, there are two solutions to the quadratic formula~\eqref{eq:quadForm}. Let $e_2$ and $\tilde{e}_2$ be the two solutions to our quadratic formula. Then
\begin{equation}\label{eq:e2timese2tilde}
    e_2 \tilde{e}_2 =\frac{c}{a}\,,
\end{equation}
and
\begin{equation}\label{eq:e2minuse2tilde}
    \vert e_2 - \tilde{e}_2 \vert = \frac{\sqrt{b^2 - 4 ac}}{\vert a \vert}\,.
\end{equation}
The fact that $c$ is negative, combined with equation~\eqref{eq:e2timese2tilde}, implies that exactly one of $e_2$ or $\tilde{e}_2$ is negative. Without loss of generality, assume $e_2$ is positive. Equation~\eqref{eq:e2minuse2tilde}, in conjunction with inequality~\eqref{eq:boundona}, implies that
\begin{equation*}
2^{-2t - \log n} \vert e_2 - \tilde{e}_2 \vert \leq \sqrt{b^2 - 4 ac} \leq 2^{-2t} \vert e_2 - \tilde{e}_2 \vert\,.
\end{equation*}
Since $e_2$ is positive and $\tilde{e}_2$ is negative, 
\begin{equation}\label{eq:boundonsqrtsize}
2^{-2t + \log n} \vert e_2 \vert \leq \sqrt{b^2 - 4 ac} \leq 2^{-2t}\vert e_2 + 1\vert\,.
\end{equation}

Let $\alpha, \beta > 0$. Then it can easily be seen that
\begin{equation*}
\left\vert \sqrt{\alpha} - \sqrt{\beta} \right\vert = \frac{\alpha - \beta}{\sqrt{\alpha} + \sqrt{\beta}}\,.
\end{equation*}
Using this fact, we have
\begin{align*}
\left\lvert \sqrt{b^2 - 4ac} - \sqrt{\hat{b}^2 - 4\hat{a} \hat{c}} \right\rvert &= \frac{\vert b^2 - 4ac - \hat{b}^2 + 4\hat{a} \hat{c} \vert}{\sqrt{b^2 - 4ac} + \sqrt{\hat{b}^2 - 4\hat{a} \hat{c}}}\\
&\leq \frac{|b-\hat{b}||b+\hat{b}|+4|a-\hat{a}||c|-4|c-\hat{c}||\hat{a}|}{\sqrt{b^2 - 4ac} + \sqrt{\hat{b}^2 - 4\hat{a} \hat{c}}}\\
&\leq \frac{|b-\hat{b}|(2|b|+|b-\hat{b}|)+4|a-\hat{a}||c|-4|c-\hat{c}|(|a|+|a-\hat{a}|)}{\sqrt{b^2 - 4ac} + \sqrt{\hat{b}^2 - 4\hat{a} \hat{c}}}\,.
\end{align*}
Applying inequalities~\eqref{eq:boundonbminusb},~\eqref{eq:boundonb},~\eqref{eq:boundonaminusa},~\eqref{eq:boundonc},~\eqref{eq:boundoncminusc}, and~\eqref{eq:boundonaminusa}, this is less than
\begin{align*}
\frac{2^{14-r-3t+3\log n}}{{\sqrt{b^2 - 4ac} + \sqrt{\hat{b}^2 - 4\hat{a} \hat{c}}}}
&\leq \frac{2^{14-r-3t+3\log n}}{\sqrt{b^2 - 4ac}}\\
&\leq \frac{2^{14-r-3t+3\log n}}{2^{-2t+\log n}|e_2|}\,,
\end{align*}
by inequality~\eqref{eq:boundonsqrtsize}.
Simplifying, we have
\begin{equation}\label{eq:boundonsqrtdiff}
    \left\lvert \sqrt{b^2 - 4ac} - \sqrt{\hat{b}^2 - 4\hat{a} \hat{c}} \right\rvert<2^{14-r-t+2\log n-\log|e_2|}\,.
\end{equation}


Now, we can write
\begin{equation}\label{eq:boundone2minuse2pt1}
    |e_2-\hat{e}_2|\leq \left|\frac{1}{2a}-\frac{1}{2\hat{a}}\right|\left(|b|+\left|\sqrt{b^2-4ac}\right|\right)+\frac{1}{2a}\left(|b-\hat{b}|+\left\lvert \sqrt{b^2 - 4ac} - \sqrt{\hat{b}^2 - 4\hat{a} \hat{c}} \right\rvert\right)\\
\end{equation}
By inequalities~\eqref{eq:boundona},~\eqref{eq:boundonb}, and~\eqref{eq:boundonsqrtsize}, together with our assumption that $t\leq r-16-3\log n-\log|e_2|$, we can bound the first term:
\begin{align*}
    \frac{1}{2}\frac{|a-\hat{a}|}{|a||\hat{a}|}\left(|b|+\left|\sqrt{b^2-4ac}\right|\right)&\leq \frac{1}{2}\frac{|a-\hat{a}|}{|a||\hat{a}|}\left(2^{1-2t+\log n}+2^{1-2t }\right)\\
    &\leq \frac{1}{2}\frac{2^{4-r - t}}{|a||\hat{a}|}\left( 2^{2-2t+\log n}\right)\\
    &\leq \frac{2^{5-r - 3t+\log n}}{|a|(|a|-|a-\hat{a}|)}\\
    &\leq \frac{2^{5-r - 3t+\log n}}{2^{-2t-\log n}(2^{-2t-\log n}-2^{4-r-t})}\\
    &\leq \frac{2^{5-r - 3t+\log n}}{2^{-2t-\log n}(2^{-1-2t-\log n})}\\
    &=2^{6-r+t+3\log n}\,.\stepcounter{equation}\tag{\theequation}\label{eq:boundone2minuse2pt2}
\end{align*}
By inequalities~\eqref{eq:boundonbminusb},~\eqref{eq:boundona}, and~\eqref{eq:boundonsqrtdiff}, we can bound the second term:
\begin{align*}
    \frac{1}{|2a|}\left(|b-\hat{b}|+\left\lvert \sqrt{b^2 - 4ac} - \sqrt{\hat{b}^2 - 4\hat{a} \hat{c}} \right\rvert\right) &< \frac{1}{|2a|}\left(|b-\hat{b}|+ 2^{14-r-t+2\log n-\log|e_2|} \right)\\
     &< \frac{1}{|2a|}\left(2^{5-r-t-\log n}+ 2^{14-r-t+2\log n-\log|e_2|} \right)\\
     &\leq  2^{2t + \log n}\left(2^{5-r-t-\log n}+ 2^{14-r-t+2\log n-\log|e_2|} \right)\\
     &\leq  2^{15-r+t+3\log n-\log|e_2|}.\stepcounter{equation}\tag{\theequation}\label{eq:boundone2minuse2pt3}
\end{align*}
Combining inequalities~\eqref{eq:boundone2minuse2pt1},~\eqref{eq:boundone2minuse2pt2}, and~\eqref{eq:boundone2minuse2pt3} gives us
\begin{align*}
    |e_2-\hat{e}_2|&\leq 2^{6-r+t+3\log n}+2^{15-r+t+3\log n-\log|e_2|}\\
    &\leq 2^{16-r+t+3\log n-\log|e_2|}\,. \stepcounter{equation}\tag{\theequation}\label{eq:boundone2minuse2}
\end{align*}


We can choose the input parameter $h_1$ such that $\hat{e}_1$ has the same sign as $e_1$, so by inequality~\eqref{eq:boundone2minuse2} and our assumption that $t\leq r-16-3\log n-\log|e_2|$,
\begin{align*}
    |e_1-\hat{e}_1|&=\frac{|e_1^2-\hat{e}_1^2|}{|e_1+\hat{e}_1|}\\
    &\geq \frac{1}{|e_1|}|e_1^2-\hat{e}_1^2|\\
    &=\frac{1}{|e_1|}\left|1-e_2^2-\sum_{k=3}^n e_k^2-\left(1-\hat{e}_2^2-\sum_{k=3}^n d_k^2\right)\right|\\
    &\leq\frac{1}{|e_1|}\left( |e_2-\hat{e}_2||e_2+\hat{e}_2|+\sum_{k=3}^n|e_k-d_k||e_k+d_k|\right)\\
    &\leq\frac{1}{|e_1|}\left(|e_2-\hat{e}_2|(|e_2|+|e_2-\hat{e}_2| )+\sum_{k=3}^n|e_k-d_k|(|e_k|+|e_k-d_k|)\right)\\
    &\leq\frac{1}{|e_1|}\left(2^{16-r+t+3\log n-\log|e_2|}\left(1+2^{16-r+t+3\log n-\log|e_2|}\right)+\sum_{k=3}^n 2^{-r}(1+2^{-r})\right)\\
    &\leq\frac{1}{|e_1|}\left(2^{17-r+t+3\log n-\log|e_2|}+2^{1-r+\log n}\right)\\
    &\leq 2^{18-r+t+3\log n-\log|e_2|-\log|e_1|}\,.\stepcounter{equation}\tag{\theequation}\label{eq:boundone1minuse1}
\end{align*}

Finally, we use inequalities~\eqref{eq:boundone2minuse2} and~\eqref{eq:boundone1minuse1} to bound $\|e-\hat{e}\|$:
\begin{align*}
    \|e-\hat{e}\|&\leq |e_1-\hat{e}_1|+|e_2-\hat{e}_2|+\sum_{k=3}^n|e_k-d_k|\\
    &< 2^{18-r+t+3\log n-\log|e_2|-\log|e_1|}+2^{16-r+t+3\log n-\log|e_2|}+2^{-r+\log n}\\
    &< 2^{19-r+t+3\log n-\log|e_2|-\log|e_1|}\,,
\end{align*}
so letting $\alpha=19-3\log n-\log|e_2|-\log|e_1|$ completes the proof.
\end{proof}

\end{document}